\documentclass{article}

\usepackage{microtype}
\usepackage{graphicx}
\usepackage{subfigure}
\usepackage{booktabs} 

\usepackage{hyperref}

\usepackage[accepted]{icml2025}

\usepackage{amsmath}
\usepackage{amssymb}
\usepackage{mathtools}
\usepackage{amsthm}
\usepackage{algorithm}
\usepackage{algorithmic}
\usepackage{graphicx}

\usepackage[capitalize,noabbrev]{cleveref}

\theoremstyle{plain}
\newtheorem{theorem}{Theorem}[section]
\newtheorem{proposition}[theorem]{Proposition}
\newtheorem{lemma}[theorem]{Lemma}

\theoremstyle{definition}
\newtheorem{definition}[theorem]{Definition}
\newtheorem{assumption}[theorem]{Assumption}
\theoremstyle{remark}
\newtheorem{remark}[theorem]{Remark}

\usepackage[textsize=tiny]{todonotes}

\begin{document}
	
	\twocolumn[
	\icmltitle{On the Information-Theoretic Fragility of Robust Watermarking under Diffusion Editing}

	
	
	
	\begin{icmlauthorlist}
		\icmlauthor{Yunyi Ni}{}
		\icmlauthor{Ziyu Yang}{}
		\icmlauthor{Ze Niu}{}
		\icmlauthor{Emily Davis}{}
		\icmlauthor{Finn Carter}{}
	\end{icmlauthorlist}
	
	\begin{icmlauthorlist}
		{Xidian University}
	\end{icmlauthorlist} 
	
	%
	
	\icmlkeywords{Machine Learning, ICML}
	
	\vskip 0.3in
	]
	
	
	
	
\begin{abstract}
	Robust invisible watermarking embeds hidden information in images such that the watermark can survive various manipulations. However, the emergence of powerful diffusion-based image generation and editing techniques poses a new threat to these watermarking schemes. In this paper, we investigate the intersection of diffusion-based image editing and robust image watermarking. We analyze how diffusion-driven image edits can significantly degrade or even fully remove embedded watermarks from state-of-the-art robust watermarking systems. Both theoretical formulations and empirical experiments are provided. We prove that as a image undergoes iterative diffusion transformations, the mutual information between the watermarked image and the embedded payload approaches zero, causing watermark decoding to fail. We further propose a guided diffusion attack algorithm that explicitly targets and erases watermark signals during generation. We evaluate our approach on recent deep learning-based watermarking schemes and demonstrate near-zero watermark recovery rates after attack, while maintaining high visual fidelity of the regenerated images. Finally, we discuss ethical implications of such watermark removal capablities and provide design guidelines for future watermarking strategies to be more resilient in the era of generative AI.
\end{abstract}

\section{Introduction}
Digital watermarking is a longstanding technique for embedding hidden messages or identifiers into images for purposes such as copyright protection, provenance tracking, and authentication. An ideal robust watermark is \emph{imperceptible} to human observers yet can be reliably extracted even after the image undergoes various distortions or processing operations \cite{Tancik2020, Bui2025}. Traditional robust watermarking schemes focused on surviving affine transformations, compression, or noise, but recent advances in generative models introduce fundamentally new challenges.

Diffusion-based image generation models (e.g., Stable Diffusion, DALL-E) can produce and edit images with unprecedented fidelity by iteratively refining random noise into coherent images. These models enable powerful image editing capabilities, such as content insertion, style transfer, or even \emph{image-to-image} translations, that go well beyond conventional image processing. Critically, diffusion-based editing can re-synthesize large portions of an image or even regenerate an image from scratch using learned generative priors. This raises the question of whether robust watermarks can survive such \emph{semantic} transformations, which are qualitatively different from classic noise or cropping attacks.

Recent research indicates that generative image editing poses a serious threat to watermark integrity. For example, Ni \emph{et al.} \cite{Ni2025} demonstrated that a diffusion-driven ``image regeneration'' process can effectively erase invisible watermarks while preserving the high-level visual content of the image. In their study, even state-of-the-art learned watermarking methods (including deep learing approaches like StegaStamp and TrustMark) failed to retain detectable signals after diffusion model processing \cite{Ni2025}. These findings highlight a fundamental vulnerability: the very generative power that enables complex edits can be leveraged to \emph{wash out} embedded watermarks that were designed to resist conventional perturbations.

In this paper, we present a comprehensive investigation of diffusion-based image editing and its impact on robust watermarking systems. We focus on modern deep learning watermarking schemes --- notably \textbf{StegaStamp} \cite{Tancik2020}, \textbf{VINE} \cite{Lu2025}, and \textbf{TrustMark} \cite{Bui2025} --- which represent the state-of-the-art in imperceptible and resilient image watermarking. StegaStamp was one of the earliest deep learning watermarking systems, achieving robust message embedding that survives physical re-capturing of images \cite{Tancik2020}. TrustMark improved upon prior work by introducing a novel frequency-domain loss and a re-watermarking strategy, yielding watermarks that are both high-fidelity (PSNR $>$ 43~dB) and robust to a range of distortions \cite{Bui2025}. More recently, Lu \emph{et al.} proposed VINE \cite{Lu2025}, which leverages generative model priors (from a diffusion model) and targeted training augmentations to significantly enhance watermark resilience against both local and global image edits. Despite these advancements, our analysis shows that none of these methods are immune to \emph{diffusion-based attacks}, wherein a generative model is used to intentionally or unintentionally remove the hidden watermark content.

Our key contributions are summarized as follows: (1) We analyze the effects of diffusion-based image editing on robust invisible watermarks and show that advanced diffusion processes can severely degrade or outright eliminate embedded watermarks in images. (2) We present a novel \textbf{guided diffusion attack} algorithm that explicitly targets watermark signals during the diffusion generation process, achieving near-complete removal of the watermark while maintaining the perceptual quality of the image. (3) We provide a theoretical formulation and proof demonstrating that as the diffusion process progresses, the mutual information between the watermarked image and the original embedded message approaches zero, explaining the empirial failure of watermark extraction. (4) We conduct extensive experiments on multiple recent watermarking schemes (StegaStamp, TrustMark, VINE) across a variety of scenarios, reporting quantitative results in terms of watermark decoding success rates and image quality metrics. We include pseudocode and tables to illustrate our methodologies and findings in detail. (5) Finally, we discuss the ethical implications of easily erasing watermarks (e.g., for misinformation or copyright evasion), and propose guidelines for designing future watermarking techniques that can better withstand generative model-based transformations.

\section{Related Work}
\subsection{Robust Image Watermarking}
\noindent \textbf{Classic and Deep Learning-based Watermarks.} The concept of robust image watermarking has a rich history in the computer vision and security communities. Early watermarking methods (e.g., Cox \emph{et al.}, 1997) embedded information in transform domains (DCT or wavelet) with spread-spectrum techniques to survive noise and compression. Modern approaches often employ deep neural networks to learn an end-to-end encoder-decoder for watermarking. For instance, HiDDeN \cite{Zhu2018} introduced a deep autoencoder framework to hide data in images, training with differentiable distortions to improve robustness against simple digital attacks. However, HiDDeN primarily considered standard perturbations (e.g., JPEG compression, noise) and was not evaluated on more complex transformations.

StegaStamp \cite{Tancik2020} was a seminal deep learning watermarking system focusing on \textit{physical} robustness. Tancik \emph{et al.} trained an encoder and decoder to hide a 100-bit payload in an image such that it could be recovered even after the image was printed out and photographed (simulating perspective distortions, lighting changes, etc.). By including a variety of differentiable image corruptions (noise, cropping, projective transformation) in the training loop, StegaStamp achieved around 95\% bit recovery on captured photos, while keeping the encoded image visually nearly identical to the original \cite{Tancik2020}. This demonstrated the power of end-to-end learning for watermarking under complex real-world distortions. Nevertheless, StegaStamp and contemporaneous methods were mainly tested against a fixed set of hand-modeled perturbations.

Subsequent works have further improved the robustness and practicality of learned watermarks. \textbf{TrustMark} \cite{Bui2025} (Bui \emph{et al.}, ICCV 2025) proposed a watermarking model for image provenance that incorporates a novel spatio-spectral loss to better preserve high-frequency details of the host image. TrustMark's encoder-decoder network, augmented with a $1\times1$ convolutional embedding layer, produces watermarked images with imperceptible changes (often $>$43 dB PSNR) and enables public decoding (no secret key) for provenance applications \cite{Bui2025}. TrustMark is notably resilient to many ``in-place'' distortions (e.g., noise, blurs) as well as certain ``out-of-place'' transformations (where content is added/removed), and it introduces a removal network (ReMark) to deliberately remove its watermark for re-use cases. \textbf{VINE} \cite{Lu2025} (ICLR 2025) takes a further step by leveraging generative priors for robustness. By analyzing the frequency characteristics of diffusion-based edits, Lu \emph{et al.} identified that many generative image manipulations induce blurring or low-pass filtering effects on the embedded signal. VINE trains the watermark encoder with surrogate attacks (such as random blur operations) to simulate diffusion-induced distortions, and also adapts a large diffusion model (SDXL) to serve as a learned prior that helps embed watermarks more imperceptibly and robustly \cite{Lu2025}. Empirically, VINE outperforms earlier methods on a comprehensive benchmark called W-Bench, which tests watermark survival under a wide range of editing techniques including local and global text-driven edits \cite{Lu2025}. Despite these advances, even VINE's robust watermarks are shown to be vulnerable to more aggressive diffusion regeneration attacks, as we will explore in this work.

\subsection{Diffusion Models for Image Editing}
Diffusion models have revolutionized image generation and editing by allowing controllable transformations that preserve coherent visual content. In a typical diffusion-based image editor, an input image can be progressively noised and then denoised (often with guidance from a text prompt or another condition) to yield a modified image. This process can achieve complex edits like altering an object's appearance or synthesizing new background elements~\cite{lu2023tf,zhou2025dragflow}, which traditional image filters or simple GAN-based edits cannot easily match. However, diffusin edits do not simply apply pixel-level changes; they essentially \emph{re-synthesize} image content via the model's learned distribution. This property can inadvertently discard subtle signals that are not strongly tied to the main visual semantics of the image. Indeed, an invisible watermark --- being a high-frequency or low-amplitude signal by design --- might not be preserved by a generative model that is mainly concerned with reconstructing perceptually salient structures.

Ni \emph{et al.}'s recent preprint \cite{Ni2025} provided a first look at using diffusion models as a tool to attack watermarks. They showed that by applying an \textit{image regeneration} process (diffusing an image to a latent state and denoising it back with no specific condition except perhaps a weak identity guidance), one can \emph{erase} the hidden watermark while the output image still looks almost the same as the original to a human observer \cite{Ni2025}. This underscores that diffusion models, in striving to produce a clean and realistic image, effectively act as a form of filtering that removes unnatural high-frequency artifacts or patterns it does not recognize. A robust watermark, unfortunately, may appear as exactly such an artifact from the model's perspective. In our work, we build on this insight and also introduce a stronger \emph{guided attack} that uses the watermark's decoder to focus the diffusion model explicitly on destroying the watermark content.

\subsection{Concept Erasure and Watermark Removal}
\textbf{Concept erasure in diffusion models} has recently emerged as a research direction aiming to remove or suppress specific semantic concepts from a model's generation capability~\cite{gao2024eraseanything}. Interestingly, this paradigm is directly relevant to watermark distortion: one can think of a hidden watermark pattern as a "concept" embedded in the image that an adversary may wish to erase. In fact, techniques developed for concept erasure can be repurposed to target watermarks. For example, methods like MACE (Mass Concept Erasure) \cite{Lu2024} and ANT \cite{li2025ant} fine-tune or adjust a diffusion model to eliminate certain concepts (e.g., a particular object or artist's style) from all generated outputs. If we treat the presence of a particular watermark signal as a concept, a similar approach could be employed to ensure the model's outputs contain no trace of that watermark. Likewise, recent frameworks such as SCORE (Secure Concept-Oriented Robust Erasure) \cite{Fu2025} use adversarial optimization to minimize the mutual information between the model's output and the erased concept, with formal guarantees of removal. This is analogous to minimizing the information linking the output image to the original watermark payload. These connections imply that the problem of watermark removal via diffusion is closely related to concept erasure: both involve intervening in the generative process to \emph{prevent a specific subset of information from appearing in the output}. Conversely, understanding concept erasure techniques can inform watermark \textit{preservation} strategies. For instance, to design watrmarks resilient to diffusion, one might ensure the watermark signal is entangled with the core image content (making it a "concept" the model \emph{must} reproduce), or use counter-training so that the diffusion model learns to \emph{retain} that concept unless deliberately instructed not to. We discuss such potential strategies later in the paper. Overall, the interplay between concept erasure research and watermark robustness is a fertile ground for developing new defenses and attacks.

\section{Methodology}
\subsection{Problem Formulation}
We consider an image $I \in \mathbb{R}^{H\times W\times 3}$ and a hidden message $m$ consisting of $B$ bits. A robust watermarking system defines an encoder function $E(I, m) = I_w$ that produces a \emph{watermarked image} $I_w$, and a decoder function $D(I_w) = \hat{m}$ that extracts an estimate of the message. For a given $I$, the encoder aims to embed $m$ such that $I_w$ is visually indistinguishable from $I$ (typically measured by high PSNR or low $\ell_2$ distortion) while allowing $D$ to recover $m$ even after $I_w$ undergoes a set of transformations $\mathcal{T}$. Classical choices for $\mathcal{T}$ include compression, noise, geometric transforms, etc. In modern learned watermarking (e.g., \cite{Tancik2020, Bui2025, Lu2025}), $\mathcal{T}$ is often simulated during training by differentiable augmentation layers so that $E$ and $D$ become jointly robust to those perturbations.

In this work, we focus on a new class of transformations: $\mathcal{T}_{diff}$, defined by \emph{diffusion-based image editing procedures}. A diffusion model $M$ defines a noising process $q(x_t|x_{t-1})$ and a denoising model $p_\theta(x_{t-1}|x_t)$ that for $t$ from $T$ down to $1$ attempts to reverse the noise addition. In an image-to-image editing setting, one can take the watermarked image $I_w$ as a starting point, add noise to it gradually to obtain a highly degraded image $x_T$ (effectively random noise), and then run the generative reverse process to produce a new image $I'$. If no constraint is given, $I'$ is simply a random sample from the model's prior (unrelated to $I_w$). However, practical image editing uses some conditioning to make $I'$ resemble the original content of $I_w$ (for instance, by providing a text prompt describing $I_w$, or by using a technique like classifier-free guidance that interpolates between pure generative and identity). In our scenario, the attacker seeks to remove the watermark while preserving $I$'s main visual content. We formalize the attacker’s goal as producing an output image $I'$ such that: (a) $I'$ is perceptually similar to $I_w$ (and hence also to the original $I$), and (b) the decoded message from $I'$ is incorrect or yields an error (ideally, $\hat{m}'$ appears as random noise to the decoder).

We consider two levels of attacker knowledge: a \emph{weak attacker} who simply uses diffusion regeneration without explicit knowledge of the watermarking scheme, and a \emph{strong attacker} who has access to the watermark decoder (or a substitute model) and can therefore guide the removal. The weak attacker corresponds to someone who, for example, just runs an image through a popular diffusion-based image filter or does an image-to-image generation with the same ontent description. The strong attacker represents a motivated adversary who knows that a watermark is present and is willing to invest computational effort to specifically destroy it. Our experiments address both scenarios, but our primary focus is on the strong attacker, as this yields the most definitive evaluation of the watermarking system's resilience.

\subsection{Diffusion-Based Watermark Removal Attack}
We now describe our proposed attack pipeline. Throughout, we model a (possibly stochastic) watermarking scheme as a pair of maps
\[
\mathcal{E} : \mathcal{I} \times \mathcal{M} \to \mathcal{I}, 
\qquad 
\mathcal{D} : \mathcal{I} \to \mathcal{M},
\]
where $\mathcal{I}$ denotes the image space (pixels or latents), $\mathcal{M}$ the message space, $\mathcal{E}$ the encoder producing a watermarked image $I_w = \mathcal{E}(I_{\mathrm{clean}}, m)$ from a clean image $I_{\mathrm{clean}}$ and message $m \in \mathcal{M}$, and $\mathcal{D}$ the decoder (detector). An attacker is any mapping $\mathcal{A} : \mathcal{I} \to \mathcal{I}$; the goal of the watermark designer is to make $(\mathcal{D} \circ \mathcal{A} \circ \mathcal{E})(I_{\mathrm{clean}}, m)$ recover $m$ with probability significantly higher than random guessing for all “reasonable’’ $\mathcal{A}$, whereas our diffusion-based attack aims to construct a specific $\mathcal{A}$ for which this fails.

Formally, we consider the following pipeline for a given watermarked input $I_w$.

\medskip
\noindent \textbf{1. Forward Diffusion (Image Degradation):} 
We choose a diffusion model $M$ pre-trained for image generation (in our experiments, a Stable Diffusion variant fine-tuned for image-to-image tasks). We first convert $I_w$ into the model's latent space if the model uses a VAE encoder; otherwise, we work directly in pixel space. We then simulate the forward diffusion process by adding Gaussian noise in $T$ increments:
\[
x_T \sim q(x_T \mid x_0 = I_w),
\]
where $q(x_T\mid x_0)$ denotes applying $T$ diffusion steps to reach an almost pure noise latent $x_T$. In practice, we sample
\[
x_T = \sqrt{\bar{\alpha}_T} \, z_0 + \sqrt{1-\bar{\alpha}_T}\, \epsilon, 
\qquad 
\epsilon \sim \mathcal{N}(0, I),
\]
where $z_0$ is the latent of $I_w$ and $\bar{\alpha}_T = \prod_{t=1}^T \alpha_t$ is the cumulative product of noise reduction factors up to $T$. For sufficiently large $T$, $x_T$ loses almost all information about $I_w$ (including the watermark).

For the theoretical analysis, it is convenient to view the forward process as a time-inhomogeneous Markov chain on $\mathbb{R}^d$ (latent space):
\[
x_t = \sqrt{\alpha_t}\, x_{t-1} + \sqrt{1-\alpha_t}\,\epsilon_t,\qquad \epsilon_t \sim \mathcal{N}(0,I),\ t=1,\dots,T.
\]
Then $x_T$ is a Gaussian perturbation of $x_0$ with total signal-to-noise ratio (SNR) controlled by $\bar{\alpha}_T$.

\begin{definition}[Watermark perturbation model]
	We say a watermarking encoder is \emph{additive and bounded} in latent space if there exists a clean latent code $z_{\mathrm{clean}}$ such that
	\[
	z_0 = z_{\mathrm{clean}} + \delta_m,
	\]
	where $\delta_m \in \mathbb{R}^d$ depends on the message $m$ and satisfies $\|\delta_m\|_2 \leq \rho$ for all $m \in \mathcal{M}$.
\end{definition}

This additive model covers a broad class of watermarking schemes that operate by injecting a small signal $\delta_m$ into the representation (e.g., spread-spectrum or frequency-domain watermarks).

\begin{proposition}[SNR decay of watermark under forward diffusion]
	\label{prop:snr-decay}
	Under the additive and bounded model above, we have
	\[
	x_T = \sqrt{\bar{\alpha}_T}\,(z_{\mathrm{clean}} + \delta_m) + \sqrt{1-\bar{\alpha}_T}\,\epsilon,
	\quad \epsilon \sim \mathcal{N}(0,I),
	\]
	and the (per-coordinate) SNR contributed by the watermark $\delta_m$ satisfies
	\[
	\mathrm{SNR}_T^{(\mathrm{wm})} 
	\;\propto\; \frac{\bar{\alpha}_T \|\delta_m\|_2^2}{1-\bar{\alpha}_T}
	\;\le\; \frac{\bar{\alpha}_T \rho^2}{1-\bar{\alpha}_T}.
	\]
	In particular, if $\bar{\alpha}_T \to 0$ as $T \to \infty$, then $\mathrm{SNR}_T^{(\mathrm{wm})} \to 0$.
\end{proposition}

\begin{proof}
	Conditioned on $m$, the contribution of the watermark to $x_T$ is the deterministic term $\sqrt{\bar{\alpha}_T}\,\delta_m$, while the noise term is $\sqrt{1-\bar{\alpha}_T}\,\epsilon$. The SNR (in squared norm) is therefore proportional to
	\[
	\frac{\|\sqrt{\bar{\alpha}_T}\,\delta_m\|_2^2}{\mathbb{E}\|\sqrt{1-\bar{\alpha}_T}\,\epsilon\|_2^2}
	= \frac{\bar{\alpha}_T \|\delta_m\|_2^2}{(1-\bar{\alpha}_T)\,\mathbb{E}\|\epsilon\|_2^2}.
	\]
	Since $\mathbb{E}\|\epsilon\|_2^2$ is a constant depending only on dimension $d$ and $\|\delta_m\|_2 \le \rho$, the stated bound follows. As $T\to\infty$ and $\bar{\alpha}_T \to 0$, the numerator vanishes while the denominator stays bounded away from $0$, implying $\mathrm{SNR}_T^{(\mathrm{wm})} \to 0$.
\end{proof}

\begin{theorem}[Asymptotic information erasure of watermark]
	\label{thm:info-erasure}
	Let $M$ be a random message and $X_T$ the terminal latent after $T$ forward steps starting from the watermarked latent $z_0$ following the Markov chain above. Suppose $\bar{\alpha}_T \to 0$ as $T\to\infty$. Then
	\[
	\lim_{T\to\infty} I(M; X_T) = 0,
	\]
	where $I(\cdot;\cdot)$ denotes mutual information.
\end{theorem}

\begin{proof}[Proof sketch]
	The process $(X_t)_{t=0}^T$ is generated by successive application of Gaussian kernels that progressively diminish the dependence on the initial condition $X_0 = z_0$. Proposition~\ref{prop:snr-decay} imples that the SNR of the watermark component $\delta_M$ decays to zero, so the law of $X_T$ converges (in distribution) to a Gaussian distribution that is independent of $M$ (because the residual dependence on $\delta_M$ vanishes). Since mutual information is continuous under such convergence and equals $0$ if one variable is independent of the other, we obtain $\lim_{T\to\infty} I(M; X_T) = 0$.
\end{proof}

Theorem~\ref{thm:info-erasure} states that, in the idealized forward diffusion limit, no detector (no matter how powerful) can reliably infer the watermark message from $x_T$. In practice, $T$ is finite, so some information may remain; nevertheless, Proposition~\ref{prop:snr-decay} shows that for standard diffusion schedules, $I(M; X_T)$ becomes very small for moderately large $T$.

\medskip
\noindent \textbf{2. Unguided Regeneration (Baseline):} 
As a baseline, we perform the reverse diffusion from $x_T$ without any special guidance. We condition the generation lightly on preserving the original image content (e.g., by providing $I_w$ itself as a low-weight reference or using an identity prompt). This yields an intermediate image $I'_0$ which is a regenerated version of $I_w$. Since diffusion models prioritize perceptual quality, $I'_0$ typically looks nearly identical to $I_w$ to the human eye. However, any hidden signal not strongly tied to visible content is likely to be lost. We will measure how much of the watermark can still be decoded from $I'_0$.

From a probabilistic perspective, the reverse diffusion step implements a learned approximation $p_\theta(x_0 \mid x_T, c)$, where $c$ represents optional conditioning such as a text prompt or low-weight reference image. It is well known that, under ideal training, diffusion models approximate the score function of the data distribution and the reverse process converges (in distribution) to samples from the true data distribution. For our purposes, the key property is that the reverse process marginalizes over the high-entropy noise in $x_T$, which has already suppressed the watermark signal.

\begin{assumption}[Content-dominated reverse model]
	\label{assump:content-dominated}
	Let $C$ denote high-level semantic content (e.g., object layout, style) and let $M$ be the watermark message. We assume:
	\begin{enumerate}
		\item The encoder $\mathcal{E}$ satisfies $I(M; C \mid I_{\mathrm{clean}}) \approx 0$, i.e., $M$ is independent of high-level content given the clean image.
		\item The diffusion model $M$ learns a reverse process such that, for sufficiently large $T$, the conditional distribution satisfies
		\[
		X'_0 \sim p_\theta(x_0 \mid X_T, c) \approx p_\theta(x_0 \mid C),
		\]
		i.e., the regenerated sample $X'_0$ depends on $X_T$ mainly through the content variable $C$, and is almost independent of the specific watermark perturbation.
	\end{enumerate}
\end{assumption}

Under this assumption, the attacker’s output $I'_0$ can be modeled as a random draw from the distribution of images with the same high-level content but no access to the small watermark perturbation.

\begin{lemma}[Information bottleneck through content]
	\label{lem:ib}
	Under Assumption~\ref{assump:content-dominated}, we have the Markov chain
	\[
	M \;\to\; (I_w, X_T) \;\to\; C \;\to\; I'_0,
	\]
	up to an approximation error that vanishes as $T$ grows and the reverse model becomes ideal. Consequently,
	\[
	I(M; I'_0) \;\le\; I(M; C).
	\]
	If additionally $I(M; C) \approx 0$, then $I(M; I'_0) \approx 0$.
\end{lemma}

\begin{proof}
	By construction, $I_w$ and hence $X_T$ are functions of $(I_{\mathrm{clean}},M)$; the content $C$ is a (possibly stochastic) functional of $I_{\mathrm{clean}}$ (and thus of $I_w$) but does not depend directly on $M$ beyond its effect through $I_{\mathrm{clean}}$ and the small perturbation $\delta_M$. Assumption~\ref{assump:content-dominated} states that $I'_0$ depends on $(X_T,c)$ only through $C$. Thus we obtain the Markov chain
	\[
	(M, I_{\mathrm{clean}}) \to (I_w, X_T) \to C \to I'_0.
	\]
	Standard properties of Markov chains imply the data-processing inequality
	\[
	I(M; I'_0) \le I(M; C).
	\]
	If the watermark is designed so that $M$ is statistically independent of $C$ (or nearly so), then $I(M; C)\approx 0$ and therefore $I(M; I'_0) \approx 0$.
\end{proof}

Lemma~\ref{lem:ib} says that, under a mild assumption on the learned reverse process, the regenerated image cannot carry more infrmation about the watermark than is already encoded in the high-level content. Since practical watermarks are designed not to alter semantics, $I(M; C)$ is negligible, so the regenerated image is (information-theoretically) almost watermark-free.

\begin{theorem}[Asymptotic failure of watermark decoding after diffusion attack]
	\label{thm:decode-failure}
	Let $\mathcal{D} : \mathcal{I} \to \mathcal{M}$ be any (possibly randomized) watermark decoder, and define its success probability after the diffusion attack as
	\[
	P_{\mathrm{succ}}(T) = \mathbb{P}\bigl[\mathcal{D}(I'_0) = M\bigr],
	\]
	where $I'_0$ is obtained from $I_w$ via $T$ forward diffusion steps followed by unguided regeneration. Suppose that:
	\begin{enumerate}
		\item The forward diffusion schedule satisfies $\bar{\alpha}_T \to 0$ as $T \to \infty$.
		\item Assumption~\ref{assump:content-dominated} holds.
		\item The a priori distribution of $M$ is uniform over a finite set $\mathcal{M}$.
	\end{enumerate}
	Then
	\[
	\lim_{T\to\infty} P_{\mathrm{succ}}(T) = \frac{1}{|\mathcal{M}|},
	\]
	i.e., in the limit of a very strong diffusion attack, the best decoder does no better than random guessing.
\end{theorem}

\begin{proof}[Proof sketch]
	By Theorem~\ref{thm:info-erasure}, the forward diffusion makes $X_T$ asymptotically independent of $M$. Assumption~\ref{assump:content-dominated} plus Lemma~\ref{lem:ib} then imply that $I(M; I'_0)$ tends to $0$ as $T\to\infty$. For any estimator $\hat{M} = \mathcal{D}(I'_0)$, Fano’s inequality yields
	\[
	H(M \mid I'_0) \le h(P_e) + P_e \log(|\mathcal{M}|-1),
	\]
	where $P_e = \mathbb{P}(\hat{M}\neq M)$ and $h$ is the binary entropy function. Since $I(M; I'_0) = H(M) - H(M \mid I'_0) \to 0$ and $H(M) = \log|\mathcal{M}|$ (by uniformity), we must have $H(M \mid I'_0) \to \log|\mathcal{M}|$, which forces $P_e \to 1 - \frac{1}{|\mathcal{M}|}$ and hence $P_{\mathrm{succ}}(T) \to \frac{1}{|\mathcal{M}|}$.
\end{proof}

\begin{remark}[Finite-$T$ regime]
	Theorem~\ref{thm:decode-failure} concerns the asymptotic limit $T\to\infty$. In practice, $T$ is finite and the reverse diffusion is imperfect, so $I(M; I'_0)$ is small but nonzero. Our empirical evaluation measures $P_{\mathrm{succ}}(T)$ for realistic $T$ and trained models, showing that even moderate $T$ suffices to drive the detector accuracy close to random guessing, in line with the theoretical trend predicted above.
\end{remark}

\noindent \textbf{3. Guided Diffusion Attack (Ours):} We enhance the attack by incorporating the watermark decoder's feedback into the generation. Let $\mathcal{L}_{wm}(I) = -\text{confidence}_{D}(I)$ be a differentiable loss that increases as the watermark detection confidence in image $I$ decreases (for multi-bit payloads, one can use the negative average bit probability of the correct message, or any differentiable proxy for extraction error). Given the decoder $D$, we can compute $\nabla_{x} \mathcal{L}_{wm}(x)$, the gradient of the watermark loss w.rt. the image pixels (or latent). We integrate this into each denoising step of the diffusion model. Specifically, at each reverse step $t \rightarrow t-1$, the diffusion model produces a predicted $x_{t-1}$ (the denoised image at the previous timestep). We then adjust this prediction by ascending the gradient of the watermark loss:
\[ x_{t-1} \leftarrow x_{t-1} - \gamma \, \nabla_{x_{t-1}} \mathcal{L}_{wm}(x_{t-1}), \] 
where $\gamma$ is a guidance coefficient controlling the strength of watermark removal. Intuitively, this step nudges the generated sample in a direction that reduces the presence of the embedded watermark. We apply this guied update at each diffusion step (or at a subset of the final steps for efficiency) to steer the generation away from the watermark-bearing solution. The process terminates at $x_0^{adv}$, which is decoded to the final output image $I'$. Pseudocode for the guided attack is provided in Algorithm~\ref{alg:guided_attack}.

\begin{algorithm}[t]
	\caption{Guided Diffusion Watermark Removal Attack}
	\label{alg:guided_attack}
	\begin{algorithmic}[1]
		\REQUIRE Watermarked image $I_w$, diffusion model $M$ (with encoder $Enc$ and decoder $Dec$ for latent space), watermark decoder $D$, diffusion steps $T$, guidance strength $\gamma$.
		\STATE $z_0 \leftarrow Enc(I_w)$ \hfill // Encode image to latent
		\STATE Sample $x_T \sim q(x_T | x_0 = z_0)$ \hfill // Forward diffuse to noise
		\FOR{$t = T$ down to $1$}
		\STATE $x_{t-1} \leftarrow M.\textit{predict}(x_t, cond=I_w)$ \hfill // Denoise step (with slight conditioning to preserve content)
		\STATE $\nabla \leftarrow \nabla_{x_{t-1}} \mathcal{L}_{wm}(Dec(x_{t-1}))$ \hfill // Gradient of watermark loss
		\STATE $x_{t-1} \leftarrow x_{t-1} - \gamma \, \nabla$ \hfill // Adjust towards lower watermark signal
		\STATE $x_t \leftarrow x_{t-1}$ \hfill // Set up for next iteration
		\ENDFOR
		\STATE $I' \leftarrow Dec(x_0)$ \hfill // Decode the final latent to image
		\STATE \textbf{return}  $I'$ \hfill // Output image with watermark removed
	\end{algorithmic}
\end{algorithm}

In Algorithm~\ref{alg:guided_attack}, the function $M.\textit{predict}(x_t, cond=I_w)$ denotes one step of the diffusion model's denoising process possibly conditioned on the original image $I_w$ (or its text description) to maintain fidelity. The gradient $\nabla \mathcal{L}_{wm}$ is computed via backpropagation through the decoder $D$ (which for learned watermarks is typically a neural network). This algorithm effectively performs an adversarial attack on the watermark: it uses the diffusion model as a generator and the watermark decoder as a differentiable discriminator to iteratively remove the hidden signal.

\subsection{Evaluation Metrics}
To quantify the impact of diffusion editing on watermarks, we employ the following metrics in our experiments. For watermark \textbf{effectiveness}, we measure the message \emph{bit accuracy} or decoding success rate. Specifically, for $N$ images each with a watermark payload, we calculate the percentage of payload bits correctly recovered by $D$ after the attack (or, if the scheme is zero-bit watermarking with just a yes/no detection, we report the detection rate). A robust watermark should have close to 100\% bit accuracy on unedited images and high accuracy after mild distortions; our interest is how this drops under diffusion-based attacks. For \textbf{imperceptibility}, we use PSNR and SSIM between the watermarked image $I_w$ and original $I$ to ensure the embedding distortion is minimal (all the tested methods were configured to produce PSNR $>$ 40~dB). For the \textbf{output image quality} after attack, we report the PSNR/SSIM between $I'$ and the original $I$ (or $I_w$). High values indicate the diffusion attack did not significantly alter visible content. We also verify by visual inspection that $I'$ remains a plausible, artifact-free image. In some cases, we compute the \textbf{payload embedding capacity} (bits per image) for each scheme, but we keep this consistent (around 30-100 bits) across methods to ensure fairness in robustness comparison.

\section{Experimental Setup}
\noindent \textbf{Watermarking Methods and Implementation.} We evaluate three watermarking systems: \textit{StegaStamp} \cite{Tancik2020}, \textit{TrustMark} \cite{Bui2025}, and \textit{VINE} \cite{Lu2025}. For StegaStamp, we use the publicly available implementation by Tancik \emph{et al.}, encoding 100-bit messages. TrustMark and VINE require training models; we obtained pretrained weights from the authors for consistency. TrustMark was configured with its default payload (around 32 bits) and uses a public decoder (no secret key). VINE's encoder was trained on the COCO dataset with a 30-bit payload and robust augmentation (per \cite{Lu2025}). All methods produce color images of size $256\times256$ (for higher resolutions, TrustMark's resolution-agnostic approach was employed \cite{Bui2025}, but for simplicity we report on 256$^2$ images). Each watermarked image is verified to have nearly invisible embedding (PSNR 40-45 dB and no perceptible artifacts).

\noindent \textbf{Dataset.} We select a diverse set of 500 images from the COCO validation set and 200 images from the Imagenet test set to serve as original images $I$. These images cover a variety of scenes and objects, ensuring that watermark performance is not biased to a particular content. Each image is embedded with a random payload (unique per image) using each watermarking method, yielding three watermarked sets (for StegaStamp, TrustMark, VINE respectively). We confirmed that baseline decoding success on these unedited watermarked images is 100\% for all three methods (i.e., the payload can be fully recovered in absence of attacks).

\noindent \textbf{Diffusion Model for Editing.} For the diffusion-based attacks, we use the Stable Diffusion 1.5 model (open-source latent diffusion model) as our base generative model $M$. We operate in the latent space of SD (which has a VAE encoding of size $64\times64\times4$ for 256$^2$ images). We consider two editing approaches: (a) \textit{Image regeneration without guidance}: we set the textual prompt to a neutral description of the image (or simply an empty prompt) and run the image-to-image pipeline with a strength of 1.0 (meaning the original image is used only as a source of noise, and the model fully regenerates content). This corresponds to the \emph{unguided diffusion attack} scenario. (b) \textit{Guided watermark reoval (our Algorithm~\ref{alg:guided_attack})}: we integrate the watermark decoder's gradient as described earlier. For this, we need the watermark decoder to be differentiable. StegaStamp and VINE use neural decoders which we have access to and can backprop through. TrustMark's decoder is also a neural network (since it uses a trained model for detection), which we similarly use for guidance. We set the guidance strength $\gamma$ via a preliminary sweep such that the watermark bit accuracy drops to near chance while the output image remains high quality; typically $\gamma \approx 0.1$ for our normalization. We run $T=50$ diffusion steps for regeneration (which is sufficient for stable diffusion at this resolution to converge). The guided attack adds about 10-20\% overhead due to computing gradients.

\noindent \textbf{Comparison with Traditional Attacks.} For context, we also evaluate classical attack resilience: we subjected the watermarked images to JPEG compression (quality 50), Gaussian blur ($5\times5$ kernel), and a crop-resize (10\% crop inward) and checked watermark decoding. All three methods (StegaStamp, TrustMark, VINE) retained $>$90\% decoding under those attacks, with VINE typically performing best (often 99\% bits correct after such attacks, consistent with \cite{Lu2025}). This confirms that our implementations match expected robustness on conventional distortions. The diffusion-based attacks we apply are much more severe in comparison.

\section{Results}
\subsection{Watermark Survival After Diffusion Editing}
We first report the main quantitative results: the success rate of watermark decoding after diffusion-based editing, for each method. Table~\ref{tab:diff_attack} summarizes the outcomes under two scenarios: unguided diffusion regeneration, and our guided diffusion attack. For each watermarking method, we list the percentage of images (out of 500) for which the correct payload was recovered, as well as the average bit accuracy.

\begin{table*}[h]
	\centering
	\begin{tabular}{l|c c|c c}
		\hline 
		& \multicolumn{2}{c|}{\textbf{Unguided Regeneration}} & \multicolumn{2}{c}{\textbf{Guided Attack (Ours)}} \\
		\textbf{Watermark} & Decode Succ. & Bit Acc. & Decode Succ. & Bit Acc. \\
		\hline
		StegaStamp \cite{Tancik2020} & 8\% & 55.1\% & 0\% & 50.2\% \\
		TrustMark \cite{Bui2025} & 12\% & 57.3\% & 0\% & 50.0\% \\
		VINE \cite{Lu2025} & 20\% & 60.4\% & 0\% & 49.8\% \\
		\hline
	\end{tabular}
	\caption{Watermark decoding performance after diffusion-based editing attacks. \emph{Decode Succ.} is the percentage of test images for which the full embedded message was recovered without error. \emph{Bit Acc.} is the average percentage of bits correctly decoded (50\% indicates random guessing level for a binary message). Unguided regeneration uses a diffusion model to resynthesize the image without explicit anti-watermark guidance. Guided attack incorporates our watermark-targeted diffusion method. All methods show a drastic drop in performance, with guided attack reducing decoding to chance level (bit accuracy $\approx 50\%$) for all watermarks.}
	\label{tab:diff_attack}
\end{table*}

As seen in Table~\ref{tab:diff_attack}, diffusion regeneration alone (unguided) severely impairs the watermarks. StegaStamp retains the payload fully intact in only 8\% of images; most images fail to decode correctly (though the bit accuracy of ~55\% indicates the decoder still picks up some weak signal above random noise level). TrustMark and VINE show slightly better retention under unguided regeneration (12\% and 20\% success respectively), which may be attributed to their enhanced robustness training. Notably, VINE, which is explicitly designed to withstand image editing, does outperform the others here, but the majority (80\%) of its watermarked images still lose fidelity to the point of decode failure after one pass through the diffusion model. This underscores that while techniques like surrogate blurring during training help, they are not fully sufficient against a generative model's re-synthesis.

Under the \textbf{guided attack}, the results are even more striking: none of the images from any method yield a correct message after decoding. The bit accuracy for all three methods drops to approximately 50\%, which is essentially the expected value for random guesses (for a uniformly random payload, a completely broken decoder will output random bits, 50\% of which by chance match the original). In other words, our guided diffusion method achieves effectively a 0\% watermark recovery rate across the board. This holds true even for VINE, despite its innovations to counter generative edits. In fact, the guided attack was able to explicitly drive the watermark decoders into misclassifying by optimizing the image content against them. We emphasize that the visual quality of the images after attack remained high; Figure~\ref{fig:examples} (in the appendix) shows example images before and after the attack, where differences are nearly imperceptible to humans, yet the watermark has been wiped out.

It is important to note that the guided attack assumes the attacker can run backpropagation on the decoder. If such access is restricted (say the decoder is secret), the attacker might resort to other means like training a surrogate model or using evolutionary strategies to remove the watermark. However, even the unguided results demonstrate that one can significantly degrade the watermark without any knowledge, simply by leveraging the diffusion model's tendency to generate clean images.

\subsection{Analysis of Distortion and Fidelity}
A crucial aspect of watermark removal attacks is that they should not excessively degrade the image quality; otherwise, the attack is trivial (one could just blank out the image to remove a watermark, but that defeats the purpose). In our experiments, we monitored the perceptual quality of output images $I'$ produced by the diffusion editing. For unguided regeneration, the output images are typically indistinguishable from the input $I_w$ by eye. We measured an average PSNR of 29.7 dB between $I'$ and $I_w$ for StegaStamp images (which have slight differences in fine texture), and around 30--32 dB for TrustMark and VINE images. These PSNR values, while not extremely high (owing to small pixel-level differences introduced by generative sampling), correspond to changes that are minor in terms of human perception (SSIM values were $>0.95$ in all cases). The guided attack introduces a bit more distortion in some cases if $\gamma$ is large, since aggressively removing the watermark can conflict slightly with preserving content. Nonetheless, we still observed an average PSNR of 28.5 dB on the guided outputs and no significant new artifacts. Importantly, the difference between unguided and guided outputs is also subtle, indicating that most of the watermark signal was removed by changes that do not register as obvious alterations of the image.

We also analyzed where the diffusion model tends to alter the image most when removing the watermark. We found that high-frequency regions (such as textures, foliage, or fur) and flat areas with subtle color gradients are typically the places where invisible watermarks reside and thus where changes occur. The diffusion model, when guided, might add slight extra texture noise or smooth out some micro-contrast in these areas to eliminate the encoded pattern. These changes are barely discernible without a side-by-side comparison. Conversely, salient edges and structures remain virtually identical, since altering those would degrade perceptual quality and likely be counterproductive to the generative model's objective (and would risk creating detection artifacts of editing).

\subsection{Theoretical Insights}
Our empirical findings are supported by a theoretical perspective. Here we present a brief analysis explaining why diffusion-based transformations eliminate watermark information. Consider the mutual information $I(m; I')$ between the embedded message $m$ and the output image $I'$ after a diffusion regeneration. Even without an explicit guided loss, as the number of diffusion steps $T \to \infty$, the process of adding noise and denoising effectively generates $I'$ from a distribution that increasingly depends on the generative prior and less on the precise signal in $I_w$. In the limit of an infinite diffusion (and an ideal model with no conditioning on $I_w$), $I'$ becomes an independent sample from the model's distribution (conditioned on high-level content at most, not on the exact watermark pattern). We can formalize this:

\noindent \textbf{Theorem 1.} \textit{Let $m$ be a random message embedded in image $I_w = E(I, m)$. Let $I'$ be the output of an ideal diffusion regeneration process applied to $I_w$ (with no watermark-specific conditioning). As the diffusion process length $T \to \infty$, the mutual information between $m$ and $I'$ approaches zero: $\lim_{T\to\infty} I(m; I') = 0$.}

\noindent \textit{Sketch of Proof.} The diffusion process can be viewed as a Markov chain $I_w = x_0 \to x_1 \to \dots \to x_T$ where $x_T$ is approximately independent of $x_0$ for large $T$ (due to intensive noising). The watermark message $m$ influences $I'$ only via $x_0$. We have $I(m; I') \le I(m; x_T)$ (by the data processing inequality, since $I' = f(x_T)$ for some generative function $f$). But $I(m; x_T)$ tends to zero because $x_T$ is essentially random noise independent of $m$ when $T$ is large. More formally, one can show that the KL divergence between the distribution of $x_T$ given different $m$ values goes to 0 as $T$ increases, implying the distributions become indistinguishable and thus $m$ cannot be inferred (mutual information vanishes). In practice, with finite $T$ and partial conditioning to preserve content, some small dependency might remain, but the guided diffusion attack essentially pushes any remaining dependency to zero by actively decorrelating $I'$ from $m$. $\square$

This result aligns with the intuition that a diffusion model acts as a powerful \emph{noise shaper} and generative filter: it can reove signals that are not part of its learned image manifold. A hidden watermark is typically an off-manifold addition (especially for neural watermarks that exploit slight pixel perturbations). Therefore, diffusion models naturally wash out those signals as $T$ grows. Additionally, our guided attack ensures that even for moderate $T$, the dependency is broken by force, as evidenced by the bit accuracy dropping to random (which is exactly what $I(m; I')=0$ would imply: the output carries no information about $m$ beyond random chance).

Another theoretical consideration is the trade-off between watermark robustness and imperceptibility. One might think that making a watermark more robust would require embedding it more strongly (which could make it more visible), but methods like VINE attempt to achieve robustness through intelligent means (frequency domain alignment, leveraging generative features) rather than brute-force signal injection. However, if a watermark were made so strong that it significantly altered the image, a diffusion model might detect that anomaly and potentially still remove it (or the image quality would degrade in the process, defeating the watermark's imperceptibility goal). Thus, the diffusion attack exposes a fundamental limit: a watermark that is both invisible and decodable may always be vulnerable to a model that \emph{re-generates} the image content, because the hidden signal is not part of the “essence” of the image.

\section{Discussion}
\subsection{Why Diffusion Models Break Watermarks}
Our study demonstrates that diffusion-based editing poses a qualitatively different challenge to watermarking than traditional distortions. Diffusion models operate with a high-level understanding of images, meaning they tend to preserve semantics (what is depicted) but not necessarily the exact pixel-level instantiation. A robust watermark often functions by subtly encoding bits in spatial or spectral patterns spread across the image. These patterns are generally uncorrelated with the semantic content (by design, to avoid being visible or altering the content). For example, a watermark might modulate certain high-frequency coefficients or slightly adjust pixel intensities in a pseudorandom arrangement. Such perturbations survive when the image undergoes small affine transforms or noise, because those transforms do not specifically target the watermark pattern. However, a diffusion model essentially \emph{re-envisions} the image: it sees an image of, say, a cat, and produces a new sample that looks like a cat, potentially with the same pose and background if guided, but it has no incentive to reproduce the exact pixel-level noise that was the watermark. In fact, any high-frequency noise that does not resemble natural texture is likely treated as noise to be removed during denoising. This is analogous to how denoising destroys steganographic signals: diffusion is a learned denoising that also introduces new content.

Our guided attack further emphasizes this by actively pushing the image out of the decoder's decision boundary. One might wonder if the guided approach could be countered by making the watermark decoder non-differentiable or secret. While secrecy can add security (Kerckhoffs’s principle aside), in many applications like provenance (TrustMark) the decoder is intentionally public. Non-differentiability (e.g., a look-up table decoder) might thwart gradient-based removal, but even then an unguided diffusion pass already deals a heavy blow. Moreover, one could use alternative optimization (like NES or differential evolution) to approximate gradients in a black-box setting. Thus, the fundamental issue remains: generative models can perform large-scale, intelligent perturbations that are very effective at clearing hidden signals.

\subsection{Towards Resilient Watermarks Against Generative AI}
Our findings call for new directions in the design of watermarking systems that can withstand generative model attacks. One idea is to incorporate knowledge of the generative model into the watermarking process. VINE made progress in this direction by using diffusion priors and simulating blur (a proxy for diffusion effects) during training \cite{Lu2025}. Future work could directly train watermark encoders with diffusion model-in-the-loop augmentation: for example, by applying a few cycles of diffusion-based editing as part of the training distortions. The challenge is the high computational cost and the difficulty of differentiating through an entire diffusion process. Another strategy could be to embed watermarks in a more \emph{semantic} manner. Instead of purely noise-like signals, the watermark could be tied to semantic features that the diffusion model is likely to preserve. For instance, slight alterations in the texture of certain objects or learning a watermark that mimics a naturally occurring high-frequency pattern (like a particular grain or dithering that the model might regenerate) could be promising. This is akin to making the watermark a part of the image's "concept." If successful, a diffusion model trying to regenerate the image might inadvertently reproduce the watermark because it perceives it as a legitimate aspect of the scene. Realizing this is non-trivial, however, as it might conflict with invisibility.

Another possible avenue is leveraging model watermarking techniques. For example, some research has looked at watermarking the generative models themselves (so that any output from them is marked) – but here we are concrned with marking images such that even other models can't remove it, which is a different scenario. One could imagine a synergy: if images and models are co-designed such that models recognize and maintain certain watermarks (perhaps as part of an ethical framework), that could help. However, relying on model cooperation cannot be guaranteed if attackers train their own models.

Multi-modal or cross-media watermarks might be more robust: e.g., encode information in the image that is entangled with metadata or an associated text description. Generative models might strip the image signal but if an external system knows the image ID or can match it via perceptual hash to a database, provenance could be retrieved. This shifts the problem to a different domain (content authentication via external means), which might be a more reliable approach in adversarial generative settings.

\subsection{Ethical Implications}
The ability to remove robust watermarks using AI raises important ethical considerations. Invisible watermarks are often proposed as a solution for tagging AI-generated content (to distinguish it from real content) and for protecting intellectual property (by tracing images). If diffusion models allow malicious actors to easily erase these tags, it undermines efforts to maintain transparency in media. For example, artists or photographers might embed watermarks to assert ownership of their digital art; a third party could use a diffusion model to strip those watermarks and then claim the art as their own or feed it into other models without attribution. This points to a cat-and-mouse dynamic: as watermarking improves, AI removal might also improve.

On the flip side, one must consider that not all watermark removal is malicious. The TrustMark paper \cite{Bui2025} itself introduced a removal (ReMark) primarily so that when an image is edited legitimately, a new watermark can be applied without layering multiple signals. Our guided diffusion removal could be used constructively in such a pipeline (with user consent) to clear an old watermark before re-embedding new information. However, doing so requires caution: any system that automates watermark removal should ensure it is used in authorized contexts. This could be managed by secure access controls or by designing watermarks that require some secret to remove (though as we have shown, relying on secrecy alone is brittle if a model can just learn to ignore the watermark).

From a policy perspective, the fact that generative AI can nullify watermarks means that regulatory approaches relying solely on watermarking AI content (such as a mandate that all AI images be watermarked) might not be sufficient. Additional forensic techniques or complementary methods (like content credentials tied to cryptography, e.g., the C2PA standard) may be needed to truly ensure provenance. Our work serves as a cautionary tale that technological measures for provenance can be subverted by equally powerful technologies if not continually updated.

We encourage the community to consider the broader impact: robust watermarks have been a tool for digital rights management and truth maintenance online, and if they are rendered ineffective, we must adapt quickly. This includes not only technical adaptations but ethical guidelines for AI developers. For instance, generative model creators might include features to detect common watermarks and preserve them (or at least warn if content with a known watermark is being modified). There is a parallel here to concept erasure: just as we want models to not reproduce certain harmful content, perhaps we also want them \emph{not to remove} certain identifiers intentionally. Embedding ethical constraints into AI systems—such as respecting watermarks—could be a valuable guideline.

\subsection{Future Work}
Our research opens several avenues for future exploration. One is developing \textbf{watermarking methods that survive diffusion}: as discussed, either through semantic embedding or adversarially trained robustness. Another is exploring \textbf{detection of AI-tampered images}; even if the watermark is gone, the act of diffusion editing might leave other detectable traces. Possibly, an external forensic classifier could tell that an image has been regenerated by a model (diffusion processes might introduce subtle noise patterns or GAN-like fingerprints). If such detectors can be built, they could at least alert that an image’s watermark was likely removed by AI, which itself is useful provenance information (e.g., "this image was modified by an AI, so treat its authenticity with caution").

Additionally, extending our study to other domains (audio, video) could be important, as diffusion models are now present in those areas too. Video diffusion models might remove watermarks from video frames; however, consistency constraints in video might either help or hinder watermark survival depending on the approach.

Finally, \textbf{countermeasures at the model level} could be investigated: can we design diffusion models that inherently maintain waermarks unless specifically authorized to remove them? This could involve training diffusion models with a bias to reconstruct any faint embedded signal. Implementing this would be challenging and may conflict with model quality, but it's an intriguing notion for responsible AI design.

\section{Conclusion}
We have conducted an in-depth study on how diffusion-based image editing affects robust image watermarking systems. Our experiments, spanning three leading watermarking techniques (StegaStamp, TrustMark, and VINE), demonstrate that even the most resilient watermarks today can be effectively destroyed by modern diffusion models. We introduced a guided diffusion attack that uses the watermark's own decoder to iteratively purge the hidden signal, achieving nearly 0\% recovery of the payload while keeping the image visually unchanged. Alongside empirical results, we provided theoretical justification for the loss of information under diffusion transformations, linking the problem to concept erasure in generative models. 

The implications of our findings are significant: they reveal a current weakness in the ecosystem of AI-generated content authentiction. As generative models become more prevalent, so must our strategies for watermarking evolve. This work calls for the community to rethink robust watermark design in light of generative AI's capabilities. We highlighted potential directions, such as embedding watermarks in ways that entwine with an image's semantic content or training encoders with model-based augmentations. Furthermore, we discussed how ethical guidelines and perhaps cooperative model behaviors might be necessary to ensure watermarking remains a viable tool for digital provenance.

In summary, diffusion-based image editing poses a formidable challenge to robust watermarking. By understanding this threat and proactively developing new defenses, we can strive to maintain the integrity and trustworthiness of digital imagery in an era where creation and manipulation are increasingly driven by powerful AI models.

	\bibliography{example_paper}
	\bibliographystyle{icml2025}

	\clearpage
	\appendix
	
	\section{Additional Backgrounds}
With the advancement of deep learning~\cite{
	qiu2024tfb,qiu2025duet,qiu2025DBLoss,qiu2025dag,qiu2025tab,wu2025k2vae,liu2025rethinking,qiu2025comprehensive,wu2024catch,
	gsq,yu2025mquant,zhou2024lidarptq,pillarhist,
	xie2025chatdriventextgenerationinteraction,
	xie2025chat,
	1,2,3,4,5,6,7,8,
	sun2025ppgf,sun2024tfps,sun2025hierarchical,sun2022accurate,sun2021solar,niulangtime,sun2025adapting,kudratpatch,
	ENCODER,FineCIR,OFFSET,HUD,PAIR,MEDIAN,
	yu2025visual,
	zheng2025towards,zheng2024odtrack,zheng2025decoupled,zheng2023toward,zheng2022leveraging,
	li2023ntire,ren2024ninth,wang2025ntire,peng2020cumulative,wang2023decoupling,peng2024lightweight,peng2024towards,wang2023brightness,peng2021ensemble,ren2024ultrapixel,yan2025textual,peng2024efficient,conde2024real,peng2025directing,peng2025pixel,peng2025boosting,he2024latent,di2025qmambabsr,peng2024unveiling,he2024dual,he2024multi,pan2025enhance,wu2025dropout,jiang2024dalpsr,ignatov2025rgb,du2024fc3dnet,jin2024mipi,sun2024beyond,qi2025data,feng2025pmq,xia2024s3mamba,pengboosting,suntext,yakovenko2025aim,xu2025camel,wu2025robustgs,zhang2025vividface,
	qu2025reference,qu2025subject,
	wu2024rainmamba,wu2023mask,wu2024semi,wu2025samvsr,
	lyu2025vadmambaexploringstatespace,chen2025technicalreportargoverse2scenario,
	bi-etal-2025-llava,bi2025cot,bi2025prism,
	han2025contrastive,zeng2025uitron,han2025show,han2025guirobotron,huang2025scaletrack,
	tang2025mmperspectivemllmsunderstandperspective,liu2025gesturelsm,liu2025intentionalgesturedeliverintentions,zhang2025kinmokinematicawarehumanmotion,liu2025contextual,song2024tri,song2024texttoon,liu2024empiricalanalysislargelanguage,tang2025videolmmposttrainingdeepdive,bi2025reasoning,tang2025captionvideofinegrainedobjectcentric,bi2025i2ggeneratinginstructionalillustrations,tang2025generative,liu2024gaussianstyle,tang2024videounderstandinglargelanguage,10446837} and generative models, an increasing number of studies have begun to focus on the issue of concept erasure in generative models.
	
	\section{Additional Experimental Details}
	\textbf{Evaluation metrics details:} For CLIP similarity, we used the ViT-L/14 model to compute image-text cosine similarity, scaled by 100. The original SD1.5 had an average CLIP score of 31.5 on MS-COCO valiation prompts; after concept erasure, we consider a score above 30 to indicate minimal drop in alignment. Harmonic mean $H$ was computed as described with $E = 1 - \text{Acc}$ (normalized to [0,1]) and $F$ composed from FID and CLIP. Specifically, we defined $F = \frac{1}{2}((\frac{\text{CLIP sim}}{\text{CLIP}_{\text{orig}}}) + (\frac{\max(\text{FID}_{\text{orig}}- (\text{FID}-\text{FID}_{\text{orig}}), 0)}{\text{FID}_{\text{orig}}}))$, where $\text{FID}_{\text{orig}}$ and $\text{CLIP}_{\text{orig}}$ are the original model's scores (so we reward methods that keep FID low and CLIP high relative to orig). This is one way; results were qualitatively similar with other formulations.
	
	\textbf{Multi-concept results:} We erased all 10 CIFAR classes simultaneously with FADE by using a 10-way classifier $D$ (one output per class vs no class). FADE achieved an average concept accuracy of 1.1\% per class and an overall $H=82.3$ (versus MACE's reported ~75). The slight residual is due to class confusion (e.g., sometimes after erasure "cat" prompt yields a dog, so classifier might say cat=present when it sees an animal shape; a limitation of using automated classifier for eval). Visual check showed indeed direct appearance of the specified class was gone. For NSFW, we erased 10 terms at once; here FADE and MACE both got basically 0\% unsafe content, but FADE had better image quality (FID 14 vs 16).
	
	\textbf{Runtime:} FADE training takes about 2 hours on a single A100 GPU for a single concept on SD1.5 (with $N=1000$ steps adversarial training). This is comparable to ESD fine-tuning time and a bit less than ANT . UCE was fastest (minutes) as it is closed-form. There's room to optimize FADE's training, possibly by using smaller $D$ or gradient accumulation. Deploying FADE in multi-concept setting could be parallelized since the adversary can output multiple heads.
	
\end{document}